\documentclass[letterpaper,twocolumn,twoside]{IEEEtran}
\usepackage{amsmath, amssymb, epsfig, psfrag}
\usepackage{multirow}

\def\beq{\begin{equation}}
\def\eeq{\end{equation}}
\def\beqa{\begin{eqnarray}}
\def\eeqa{\end{eqnarray}}
\def\beqan{\begin{eqnarray*}}
\def\eeqan{\end{eqnarray*}}

\def\R{{\mathbb{R}}}

\def\x{\times}

\newtheorem{theorem}{Theorem}
\newtheorem{lemma}{Lemma}

\def\SR{\mbox{\small \sffamily SR}}
\def\SRTX{\mbox{\small \sffamily SRTX}}

\def\MRX{\mbox{\small \sffamily MRX}}

\newcommand{\pbf}{\mathbf{p}}

\newcommand{\sbf}{\mathbf{s}}

\newcommand{\tbf}{\mathbf{t}}

\newcommand{\vbf}{\mathbf{v}}
\newcommand{\wbf}{\mathbf{w}}

\newcommand{\Abf}{\mathbf{A}}

\newcommand{\Gbf}{\mathbf{G}}

\newcommand{\Wbf}{\mathbf{W}}

\title{Femto-Macro Cellular Interference Control
with Subband Scheduling and Interference Cancelation}
\author{Sundeep Rangan, ~\IEEEmembership{Member,~IEEE}
\thanks{S. Rangan (email: srangan@poly.edu) is with
        Polytechnic Institute of New York University, Brooklyn, NY.}
}

\begin{document}


\setcounter{page}{1}

\maketitle
\begin{abstract}
A significant technical challenge in deploying femtocells is
controlling the interference from the underlay of femtos
onto the overlay of macros.
This paper presents a novel interference control method
where the macrocell bandwidth is partitioned into subbands,
and the short-range femtocell links adaptively allocate their
power across the subbands based on a load-spillage power control method.
The scheme can improve rate distribution in the macro network while also providing
opportunities for short-range communication as well.
Moreover, the proposed scheme requires
minimal interference coordination communication between the femtos
and macros, which is one of the main challenges in femtocell systems.
Also, simulations show certain advantages over simpler
orthogonalization schemes or power control schemes without subband
partitioning.  Further modest gains may also be possible with
interference cancelation.
\end{abstract}

\begin{keywords} Femtocells, interference coordination, cellular systems,
wireless communications, interference cancelation, fractional frequency reuse.
\end{keywords}


\section{Introduction}

Femtocells are small wireless access points that are typically
installed in a subscriber's premises,
but operate in a cellular provider's licensed spectrum.
Since femtocells can be manufactured at a very low cost,
require minimal network maintenance by the operator
and can leverage the subscriber's backhaul,
femtocells offer the possibility of
expanding cellular capacity at a fraction of the
cost of traditional macrocellular deployments.
With the surge in demand for wireless data services,
femtocells have thus attracted considerable recent attention,
both in cellular standards bodies such as the
3GPP (Third Generation Partnership Program)
\cite{3GPP25467,3GPP22220,3GPP25820}
and academic research \cite{ChaAndG:08,YehTLK:08}.

However, one of the key technical challenges in deploying
femtocells is the interference between the underlay of
small femtocells and the
overlay of comparatively large macrocells \cite{ChaAndG:08,LopezVRZ:09,Claussen:08}.
This \emph{cross-tier} interference problem differs from interference
problems in traditional macrocellular networks in several important aspects:
First, due to \emph{restricted association}, mobile terminals
(called user equipment or UEs in 3GPP terminology)
may not be able to connect
to a given femtocell even when it provides the closest
serving base station.
Such restrictions can result in strong interference both
from the macro UE transmitter onto the
femto uplink and from the femto downlink
onto the macro UE receiver.
Second, since femtocells are generally not connected directly into
the operator's core network, and instead go through the subscriber's
private ISP, backhaul signaling for interference
coordination is often more limited.
In addition, over-the-air signaling may also be difficult due to the
power disparities between the femtos and macros,
and the large number of femtos per macro.
Finally, femtos are usually deployed in an \emph{ad hoc} manner
as opposed to carefully planned deployments in macrocellular networks.
Such deployments further degrade interference
and also require that
any interference coordination be self-configuring and adaptive.

To address these problems,
this paper proposes a novel method for macro-femto interference control
based on subband scheduling, power control and, if available, interference cancelation.
In the proposed method, both the uplink (UL) and downlink (DL)
bands of the macrocell overlay network are first partitioned into \emph{subbands}.
Using standard fractional frequency reuse (FFR) methods
\cite{3GPPR1050507:05,LeiZZY:07,HanPLAJ:08,StoVis:08},
the macro transmit power is varied across the subbands to create
different reuse patterns and interference conditions in each subband.
It is well-known that this subband partitioning can
improve the rate distribution in the macro network.

Following the observations of \cite{OhCCL:10},
the basic thesis of this paper is that subband partitioning can also 
create improved opportunities for femtocell communication as well.
Specifically, we argue that, when the macrocell network employs subband partitioning,
most femto links will be able to transmit in one or more subbands
with minimal interference into the macro network.  
Moreover, using a \emph{load-spillage} power control method similar
to \cite{HandeRCW:08}, we provide a practical method by which femto
transmitters can identify the interference effect in each subband 
and optimize their power distribution across the subbands appropriately.

There are several appealing features of the proposed
interference control method:

\begin{itemize}
\item \emph{Minimal cross-tier communication:}
The proposed method requires no explicit
cross-tier communication between the femto and macrocells.
As discussed above, femto-macro coordination communication
is logistically difficult.
The proposed load-spillage method avoids such communication
and simply requires that the macro base
stations measure the total noise rise due to the femtos and
broadcast a \emph{load} factor.
The femtocell transmitters read the load
factor and adjust their transmit power
according to a simple rule.

\item \emph{Frequency reuse}:
As described in \cite{OhCCL:10},
the subband partitioning enables the spectrum
to be reused efficiently between the femto and macrocells.
Our simulations reveal that, when combined with
optimal power allocation across the
subbands, the subband partitioning can result in significant
additional throughput in the same spectrum as the macro.
This throughput can be higher than orthogonalization
between the femto and macro layers, and -- at least in the UL --
higher than the throughput with
simple reuse 1 without subband scheduling.

\item \emph{Improved throughput with interference cancelation:}
Finally,  we show in simulations that modest
throughput gains may also be possible when the femto receivers
are capable of interference cancelation via joint detection of the
desired signal and interfering macro signal.  
\end{itemize}

\subsection{Previous Work}
A general introduction on femtocells,
including descriptions of the femto-macro interference problem,
can be found in the recent
text \cite{ZhangRoche:10}, as well as several general papers
\cite{ChaAndG:08,YehTLK:08,Claussen:08,LopezVRZ:09}.
To address the femto-macro interference problem,
a number of works have considered power control and interference
mitigation techniques.
In addition to the works above,
the paper \cite{ShiReed:07}, for example,
studies adaptive attenuation
and power control methods for UMTS femtos, while
\cite{ChandAndMSG:09,JoMMY:09} study more sophisticated
power control methods based on distributed utility maximization.
These power control methods have the benefit that they
can be largely implemented in existing cellular
systems, and based on simulations such as the Femto Forum study
\cite{FemtoForum:10}, appear to work well in a range
of circumstances.
In addition, power control admits precise analyses.
For example, the works \cite{KishoreGPS:01,KishoreGPS:05,ChaAnd:09a,ChaAnd:09b}
provide analytic expressions for the outage capacity
under power control with various cell selection methods.

The contribution of this paper is to combine
these femtocell power control methods with two other
concepts:
subband partitioning and interference cancelation.

Subband partitioning, and
the related concept of fractional frequency reuse (FFR),
are well-known in the cellular industry.
Some early descriptions can be found in
\cite{3GPPR1050507:05,LeiZZY:07}
with further analyses for distributed algorithms in
\cite{HanPLAJ:08,StoVis:08}.
FFR has attracted considerable recent attention since it
is particularly easy to implement in OFDMA systems such as in
3GPP's long-term evolution (LTE) standard \cite{Dahlman:07}.

However, the application of FFR to the femto-macro
interference problem considered here requires some modifications.
FFR methods such as
\cite{3GPPR1050507:05,LeiZZY:07} employ a static power allocation
across the subbands.  Static allocations may work well in the macrocellular
network where the interference conditions between cell sites are stable.
However, short range femto links can have arbitrary and changing
locations relative to the macro and thus need adaptive schemes.
Several adaptive FFR methods have also been proposed such
as \cite{HanPLAJ:08,StoVis:08}, but would require
coordination between the femto and macros in the context of the problem here.  
As discussed above,
this coordination can be logistically difficult.
The method proposed here requires no explicit coordination
other than the broadcast of the load factors from the macro receivers.
We will see (see Theorem \ref{thm:powCtrl})
that, when the interference at the femtos is dominated
by the macro transmitters, the proposed scheme remains optimal.

Our proposed method also has benefits the
femto-macro FFR method in \cite{OhCCL:10}.
That work uses the same partitioning discussed in
Section \ref{sec:subbband}, but with the femto assigned to a single
subband based on the DL SINR.  The method proposed here permits
the femto UE or BS
to transmit on all subbands \emph{simultaneously}
Also, the DL SINR metric used in \cite{OhCCL:10} is indirect measurement
that is only statistically correlated
to the interference effect of a femto's transmission.  This work
uses a load-spillage mechanism of \cite{HandeRCW:08} 
to determine the precise interference effect in each subband and allocate 
the power across the subbands appropriately.

The second main novel component of this paper
is interference cancelation (IC) via joint detection of the
desired femto signal and interfering macro signal.
The theoretical benefits of interference cancelation and joint detection
are well-known \cite{Ahlswede:71,CoverT:91}.
Also, due to improved computational resources,
successive IC (SIC)
has begun to be implemented in commercial cellular systems
for both macrocells \cite{Andrews:05,BourdreauPGCWV:09}
and femtocells \cite{ShiReed:07}.
However, for cellular systems, SIC has been generally
used in the UL, where power control can be used to shape the
receive powers of multiple mobiles with a favorable profile for SIC.
What is new here is that we will show
that when subband partitioning is performed
in the macrolayer, interference conditions will naturally arise such
that IC may be used by the short-range links to cancel the
interference from the macro.  Unfortunately, we only observe gains with
joint detection and not with SIC alone.

It should be pointed out that a limitation in this work
is that cell selection between femtos and macros is fixed.
In reality, femto-macro cell selection is itself an important problem
and typically depends on loading, mobility and whether
the femtos employ open or closed access.  The reader is referred to
\cite{YeuNan:96,JabFuh:97,KleinHan:04} for a discussion
of these issues.

\section{Macrocell Subband Partitioning}
\label{sec:subbband}

The proposed macro-femto interference control method
assumes subband partitioning of the macro overlay network
similar to that of Oh \emph{et.\ al.} in \cite{OhCCL:10}.
In subband partitioning, both the UL and DL bands of the macro
overlay network are partitioned into \emph{subbands},
which are contiguous non-overlapping frequency intervals.
For the macro-femto interference control method, we consider the simple
subband partitioning shown in Fig.\ \ref{fig:macroSubband}.
More general partitions can also be considered.

\begin{figure}
\begin{center}
  \includegraphics[width=2.5in]{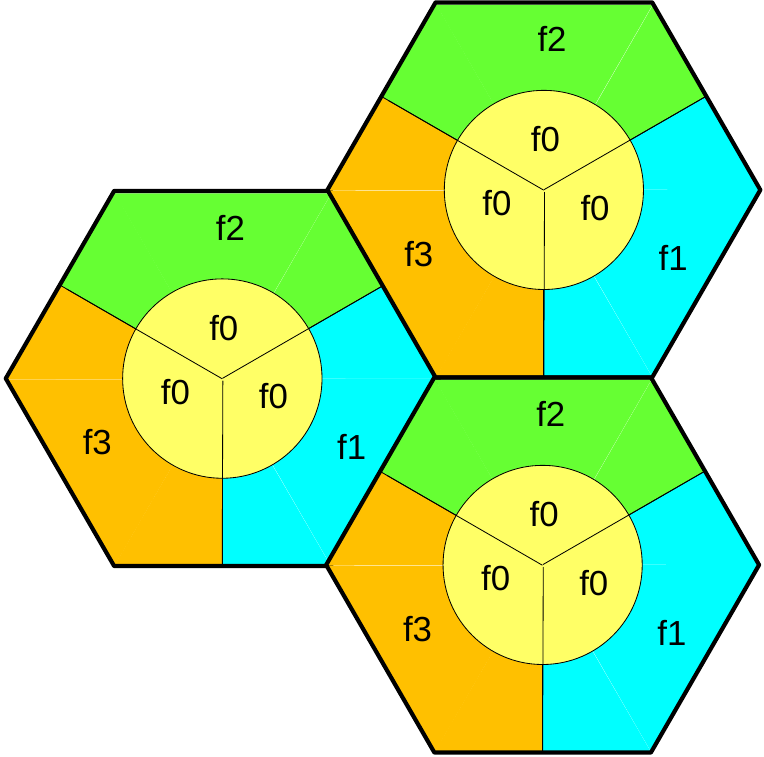}
\end{center}
\caption{Subband partitioning of the macrocells.  Each cell site
occupies one hexagon with the cell site partitioned into three 120
degree sectors called cells.  The frequency subband $f_0$ is
used in all cells and serves mobiles close to the
cell.  The other subbands $f_1$ to $f_3$ are used in only
cell per cell site and serves mobiles close to the cell edge. }
\label{fig:macroSubband}
\end{figure}

As shown in the figure, we assume the
cell sites in the macrocell network are
three way sectorized.  Following the 3GPP terminology, we call
each sector in the cell site a \emph{cell}.
With this sectorization,
we partition both the UL and DL bands
into four subbands whose center frequencies denoted $f_0$ to $f_3$.
With some abuse of notation, we will use $f_k$ to denote
the subband as well as its center frequency.

One of the four subbands, $f_0$, is reused in every cell,
and will be called the reuse 1 subband.
This subband will generally serve mobiles who are close
to the cell and is thus shown in the figure in the inner circle
of the cell site.
Each of the remaining subbands -- $f_1$ to $f_3$ --
is used in only one cell per cell site.  These subbands will be called
the reuse 3 subbands.  The reuse 3 subbands are arranged so that
neighboring cells do not use the same subbands, so the interference
in these subbands is low.  These subbands are used to serve the mobiles
at the edge of the cell where the mobile would experience
high interference if served in the reuse 1 subband.

\begin{figure}
\begin{center}
  \includegraphics[width=3in]{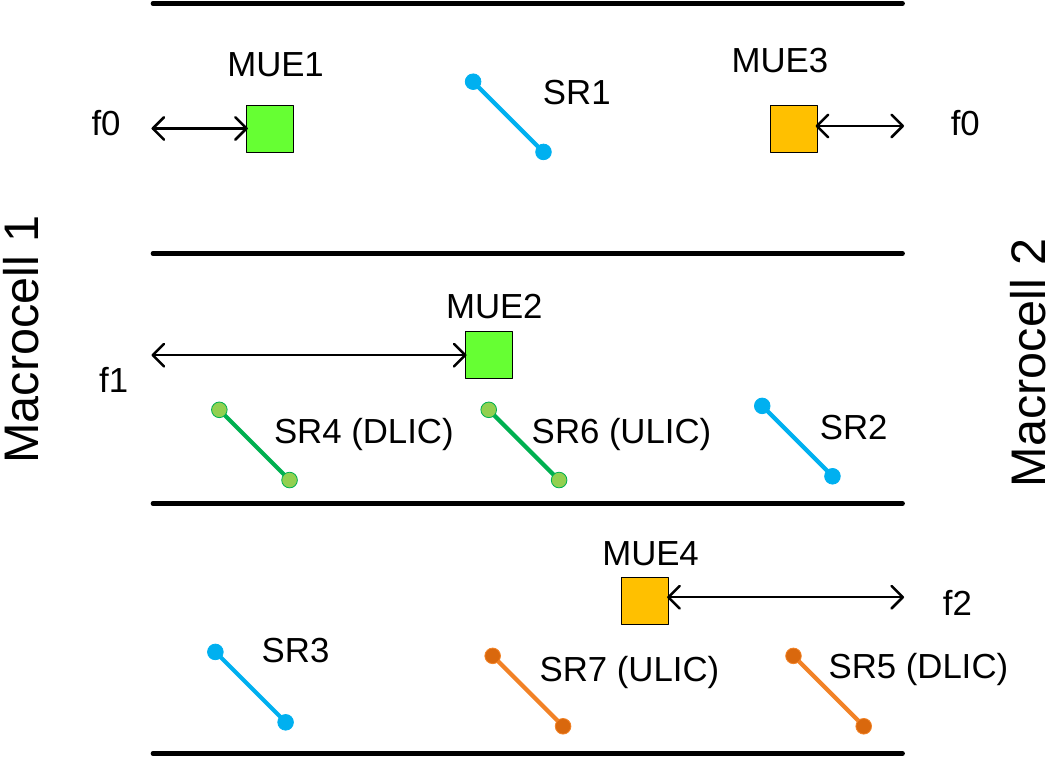}
\end{center}
\caption{Typical locations of macro UEs (MUE1--4) in each subband,
along with possible locations for short-range links (SR1--7).
SR1--3 can operate in either the UL or DL and do not require
interference cancelation (IC).  SR4 and 5 would operate in the DL
with IC of the DL signal from the macrocell base stations.
SR6 and 7 operate in the UL with IC of the macro UE transmissions.}
\label{fig:macroFemtoLinks}
\end{figure}

The basis of subband partitioning is that in any multiuser
wireless system, there is a tradeoff between maximizing
the bandwidth  (which is obtained
by reusing the spectrum in different links), or minimizing
interference (which is obtained by orthogonalizing bandwidth
in different links).
Subband partitioning enables this tradeoff to be optimized
differently for mobiles depending on their geometry.
Consider, for example, the transmissions of two macrocells
in Fig.\ \ref{fig:macroFemtoLinks}.
Macrocell 1 serves two macro UEs, MUE1 and 2, while macrocell 2
serves macro UEs MUE3 and 4.
MUE1 and MUE3 are close to their serving cells and thus
experience little interference from the other cell.
These mobiles can be served in the reuse 1 subband, $f_0$,
where the full bandwidth of the subband can be reused in both cells.
In contrast, MUE2 and 4 are at the cell edge.  If served in the same
subband, the interference would be high.  For these mobiles,
the rate may be higher by serving them on orthogonal bandwidths,
$f_1$ and $f_2$, which eliminates the interference from the other cell.

It is well-known that such subband partitioning can often improve the rate
distribution in the macrocellular network: usually the edge of cell
mobile rate can be increased
 with a relatively small loss in rate for the good mobiles.
However, what is important for this work is the observation
of \cite{OhCCL:10} that subband partitioning
also creates improved opportunities for
short-range transmissions with minimal interference into the macro network.
To see this property, consider again
Fig.\ \ref{fig:macroFemtoLinks} which also shows
three possible locations for short range links, labeled SR1 to 3.
For now, ignore the short range links SR4--7 which we will discuss
in Section \ref{sec:ic}.

First consider a short-range link, such as SR1, at the macrocell
edge (i.e.\ between the two macrocells).
A short range link such as SR1
can operate in the reuse 1 subband, $f_0$, in either the UL or DL.
For the UL, the short range link is far from either macrocell,
so the transmissions should have minimal interference at the macrocell
receivers.
Also, since the UEs in $f_0$
are typically close to their cell, they will also be far from SR1.
Therefore, transmissions on SR1 in the DL $f_0$
will also have minimal interference
at the macro UEs such as MUE1 and 3.

Now consider the cases, such as SR2 or SR3,
when the short range links are close to one of the macrocells.
In these cases, the short range links can transmit in one of the
subbands $f_1$ or $f_2$ not being used by the macrocell that it is close to.

In this way, the subband partitioning provides, wherever the short-range
link is located, frequency locations
on which it can transmit where the interference will be minimal.
This property is the basic observation of \cite{OhCCL:10} and
will serve as the basis of the proposed femto-macro interference
control method here.

\section{Subband Interference Power Control}
\label{sec:powControl}

In the method in \cite{OhCCL:10}, each femtocell transmitter
(the femto BS in the DL and femto UE in the UL)
is statically assigned to a single subband based on the DL SINR.
In the proposed method of this paper, the femtocell transmitter
can simultaneously use all subbands, with a power level in each subband
based on the precise interference conditions.
This method enables a larger fraction of the bandwidth to be used and
to select the powers in the subbands based on the actual interference
conditions.

To describe the proposed method mathematically,
suppose that there are $K$ total subbands.  For example, in the subband
partitioning of the previous section, there would be $K=4$
subbands in the UL or DL.
We will assume there are $N$ short-range links denoted $\SR_j$,
$j=1,\ldots,N$.
We will use the term short-range links, as opposed to femto links,
since the method does not require that short-range links
using the same technology as the macro.
Also, the macrocellular system could be TDD or FDD.  In the case of TDD,
the short-range links will use the UL and DL time periods
aligned with the macro.

We will assume that in subband $k$, there are
$M(k)$ macro receivers, which we will denote by
$\MRX_{ik}$, $i=1,\ldots,M(k)$.
In an UL subband, the macro receivers are the
macrocell base stations using that subband, and in the DL,
the macro receivers are the UEs using that subband.

Now let $G_{ijk}$ be the path gain from the short-range
transmitter $\SRTX_j$ to the macro receiver $\MRX_{ik}$
in subband $k$.  We will let $p_{jk}$ be the transmit
power of $\SRTX_j$ in subband $k$, so that the total interference
at $\MRX_{ik}$ in subband $k$ from all short-range links is
\beq \label{eq:qik}
    q_{ik} = \sum_{j=1}^N G_{ijk}p_{jk}.
\eeq

To implement the femto-macro interference power control,
we will assume that each short range transmitter, $\SRTX_j$,
can learn the path gains $G_{ijk}$ to all the macro receivers
in all the $K$ subbands.  To perform this estimation in the UL subbands,
the macrocell base stations would need to intermittently transmit
some reference signal in the UL.  Of course, this would require
that the regular UL macrocellular transmissions are suspended
in some quiet period and thus requires some overhead.
Similarly, in the DL, the path losses to the
macro UEs could be estimated by having a quiet period in the DL
where the macro UEs suspending their regular reception and
transmit reference signals that the short-range transmitters can hear.
Alternatively, in a TDD system,
a short range transmitter could estimate
an UL path gain to the base stations
from a DL reference such as the pilot signals,
and an DL path gain to the macro UEs from the UEs uplink power control
signals.
These reference signals can also be used in a FDD system, but with some
error due to fast fading.

Now let $\pbf_j$ be the vector of powers used by $\SRTX_j$ across the $K$
subbands,
\beq \label{eq:pvecj}
    \pbf_j := [p_{j1} \cdots p_{jK} ]^T.
\eeq
The power control problem is for each $\SRTX_j$
to select the power vector  $\pbf_j$
to control the interference levels $q_{ik}$ in \eqref{eq:qik},
while maximizing the throughput.
This is a distributed control problem since all the
short-range transmitters have some effect on the interference levels.

Under the assumptions we will make below, this power control
problem is somewhat trivial and can be solved by a large number
of distributed methods.  See, for example, the survey
\cite{ChiangHLT:08}.
As one possible solution, we present a
\emph{load-spillage} method similar to \cite{HandeRCW:08}.
Each macro receiver $\MRX_{ik}$
broadcasts a \emph{load factor}, denoted $s_{ik}$, on subband $k$.
The load factor is simply some positive scalar.
Each short range transmitter $\SRTX j$ then
computes a per subband \emph{spillage} given by
\beq \label{eq:spillage}
    r_{jk} = \sum_{i=1}^M G_{ijk}s_{jk}.
\eeq
The \SRTX $j$ can then select any set of powers $p_{jk}$ subject to
\beq \label{eq:powCons}
    \sum_{k=1}^K r_{jk}p_{jk} \leq \lambda_j,
\eeq
where $\lambda_j$ is some pre-determined positive scalar
which we will call the \emph{power constraint bound}.

The load-spillage rule \eqref{eq:powCons} can be understood
qualitatively as follows:  The rule \eqref{eq:powCons} permits
the short-range links to transmit a high power $p_{jk}$
precisely when the corresponding spillage $r_{jk}$ is small.
Now, the spillage \eqref{eq:spillage}
$r_{jk}$ is high when there is some macro receiver $\MRX_{ik}$
that is either close to the short-range transmitter (so that
$G_{ijk}$ is large) or the macro receiver
broadcasts a high load factor $s_{ik}$.
Thus, the spillage rule \eqref{eq:powCons} will cause the short-range transmitters to allocate power away from subbands
where there are close by macro receivers or macro receivers indicating
high load.

We can, in fact, mathematically show that under certain simplifying
assumptions, the load spillage rule can be optimal.
Specifically suppose that, associated with each short-range link $j$,
there is some ``utility" $U_j(\pbf_j)$ of the powers $\pbf_j$.
This utility could be, for example,
the rate or some function of the rate.  The implicit assumption
is that the utility does not depend on the power levels
in other short-range links.  That is, the interference
at the short-range receivers
is dominated by the interference from the macro transmissions.

Given such a utility function and power constraint \eqref{eq:powCons},
a natural strategy for each short-range
transmitter $\SRTX_j$ is to select the power vector $\pbf_j$ to maximize
the utility:
\beq \label{eq:Uopt}
    \max_{\pbf} U_j(\pbf_j)
\eeq
subject to \eqref{eq:powCons}.  In the case when $U_j(\pbf_j)$
is the total rate across the $K$ subbands, and the rate is given
by the Shannon capacity, this optimization reduces to a standard
waterfilling procedure \cite{CoverT:91}.

However, whatever the utility function $U_j(\pbf_j)$, we can
show the following result:
Say a selection of power vectors $\pbf = \{ \pbf_j, j=1,\ldots,N\}$
is \emph{Pareto optimal} if any other selection of
power vectors that results in a lower interference $q_{ik}$ on
on some subband $k$ for some macro receiver $\MRX_{ik}$,
or higher utility $U_j$ for some short range link $\SR_j$, must
also result in a higher interference for some other
macro receiver or lower utility in some other link.

\begin{theorem}  \label{thm:powCtrl}
Suppose that the utility functions $U_j(\pbf_j)$
are strictly increasing in each component $p_{jk}$.
Then,
\begin{itemize}
\item[(a)] If the load factors $s_{ik}$ and power constraint
bounds $\lambda_j$ are positive and each short-range link $j$
selects the power vector according to \eqref{eq:Uopt},
then the resulting power vectors are Pareto optimal.
\item[(b)] If, in addition, the utility functions are concave
and twice differentiable and $G_{ijk} > 0$ for all $i$, $j$ and $k$,
then the converse is true.
That is, if $\pbf$ is any set of power vectors that
is Pareto optimal, then there exists nonnegative $s_{ik}$ and
$\lambda_j$ such that each $\pbf_j$ satisfies \eqref{eq:Uopt}.
\end{itemize}
\end{theorem}
\begin{proof} See Appendix \ref{sec:proof}.
\end{proof}

The consequence of the theorem is that the load factors
$s_{ik}$ and $\lambda_j$ provide a parametrization
of all the Pareto optimal power vectors.  Thus,
the power vectors obtained for any $s_{ik}$ and $\lambda_j$
are guaranteed to be Pareto optimal.  Moreover,
by appropriately selecting the parameters, any optimal point
can be achieved.

However, it is important to recognize some caveats.
Most importantly, as mentioned above, the implicit assumption
in the utility function is that the power selection in
one short range link does not affect other short-range links.
This assumption is valid when the interference at the short range links
is dominated by either thermal noise or interference from the macro
transmitters.   When there is significant
interference between short-range links, the problem is
significantly more complex.
In fact, when $K > 1$, the problem is similar to the parallel channel
interference power control problem
that arises in DSL cross-talk \cite{YuGC:02}.  This problem
is generally non-convex and NP-hard.
The load spillage method can be extended to provide a suboptimal
solution to the problem.  For example, following the lines of
 \cite{CendrillonHCM:07},  the short-range receivers could themselves
 broadcast load factors and have their interference counted
 in the spillage term \eqref{eq:spillage}.
 This procedure would at best converge to some local maxima.

Also, while the theorem states that, in theory, any Pareto
optimal power vector is achievable, it provides no mechanism for
selecting the appropriate load factors $s_{ik}$ or bounds $\lambda_j$.
In the simulations below, the macro receivers will select the load factors
$s_{ik}$ simply
based on measuring the noise rise and adjusting the load factors
to keep the noise rise at a given target.  However, selection of the
bounds $\lambda_j$ or optimal selection of the load factors
$s_{ik}$ would generally require some femto-macro communication
which we are not considering.

There is indeed significant possibilities for more sophisticated
interference coordination here.
However, we will not consider such methods in this work,
as our main objective here is to show that, even simple
power control methods that use minimal coordination can
provide good performance.

\section{Interference Cancelation}
\label{sec:ic}

The proposed femto-macro interference control method
can also be used in conjunction with interference
cancelation (IC) \cite{CoverT:91}
for further throughput gains.
To see this, suppose that the short-range receiver has the computational
ability to perform IC on the macrocellular signals.
Now return again Fig.\ \ref{fig:macroFemtoLinks}
and consider first the short range links SR4 and SR5.
SR4, being close to macrocell 1, will receive the DL signal
on $f_1$ strongly.  Also, the signal on that subband will typically
be at low rate since it is intended for an edge of cell
mobile, MUE2.  Thus, with a high probability, the receiver on SR4
will be able to jointly decode and cancel the interfering signal from
macrocell 1 while receiving the signal from its desired transmitter.
Similarly, the receiver SR5 can perform IC with high likelihood
on the DL signal from macrocell 2 on $f_2$.

Now for the UL, consider short-range links SR 6 and 7.  These
links are close to the macro UEs, MUE 2 and 3.  Since MUE 2 and 3
are far from their serving macrocell, they will likely transmit
at a low rate, and thus can be canceled by the nearby short range receivers.

In this way, we see that when the short range link receivers
are capable of IC, with high likelihood, the interference
cancelation may be able to remove interference in a number of the
subbands.  Moreover, IC can be performed entirely opportunistically
in that the macro does not need to change its scheduler policy at all,
or even be aware of the short-range links.  The regular operation
of the macro under the subband partitioning will naturally open
up possibilities where IC can be used.

\section{Numerical Simulation}

\begin{table}
\begin{center}
\begin{tabular}{|p{1in}|p{2in}|}
\hline
Parameter & Value \\ \hline
Macro topology & 24 hexagonal cell sites with wrap around and
3 cells per site.  \\ \hline
Femto layout & Uniformly distributed at density of
10 per macrocell \\ \hline
Number macro UEs & 10 per macrocell. \\ \hline
Numer femto UEs & 1 per femtocell. \\ \hline
Macro cell radius & 500 m \\ \hline
Femto link distance & 20 m \\ \hline
\multirow{2}{*}{Macro antenna pattern} & $A(\theta) = -\min\left\{ 12\left(\frac{\theta}{\theta_{3dB}}\right)^2, A_m \right\}$, \\
&  $\theta_{3dB} = $ 70 degrees and $A_m =$ 20 dB \\ \hline
Femto antenna pattern & Omni \\ \hline
Lognormal shadowing (macro BS involved) & 8.1 dB,
50\% inter-site correlation;
100\% intra-site correlation
\\ \hline
Lognormal shadowing (macro BS not involved) & 4 dB, no correlation
\\ \hline
\end{tabular}
\end{center}
\caption{Parameters for simulation of the femto-macro interference power control algorithm.}
\label{tbl:simParam}
\end{table}

The method was evaluated under a simulation model
similar to Femto Forum interference
management study in \cite{FemtoForum:10} with
parameters shown in Table \ref{tbl:simParam}.
The macro network is composed of a 24
hexagonal cell sites arranged in a 4 x 6 grid with wrap
around to eliminate edge effects.  Each site has three cells
with the antenna pattern $A(\theta)$ given in Table \ref{tbl:simParam}.
The bandwidth is partitioned into reuse 1 and reuse 3 subbands
as described in Section  \ref{sec:subbband}, with a fraction
$p_1=0.5$ of the bandwidth
allocated to the reuse 1 subband.

Macro UEs are distributed uniformly with an average of 10 macro UEs per
macrocell.  We assume that the macro UEs connect to the strongest macrocell.
In each cell, the fraction $p_1$ of the macro UEs with the
highest SINR are assigned to the reuse 1 subband, and the remaining
fraction of $1-p_1$ macro UEs are assigned to the reuse 3 subband.
We assume that the macro network uses OFDMA and the
macro UEs within each subband are assigned equal, orthogonal
fractions of the subband bandwidth.
In the DL, we assume each macro BS transmits equal power density
to all the macro UEs.  In the UL, we assume power control
at each cell sets the macro UEs served in that cell
to have equal receive power density.
Moreover, we assume that the receive power densities are adjusted
so that all cells experience
an interference to thermal (IoT) of 10 dB.
In both the UL and DL, the rate is then selected
based on the received signal-to-interference
and noise ratio (SINR).  We assume that the links achieve a rate
equal to 3 dB below Shannon capacity, with a maximum spectral efficiency of
5 bps/Hz.

The underlay of femtocell BSs are randomly placed with a uniform distribution
at a density of 10 femto BSs per macrocell.
Each femto BS serves one femto UE
at a random location 20m from the femtocell BS.
Note that the femtocell links are considerably smaller than
the macro cell radius of 500m.

The path loss models are shown in Table \ref{tbl:pathLossMod}
are also based on \cite{FemtoForum:10}.  The wall loss values
are the lower values in \cite{FemtoForum:10}
under the assumption that all the macro UEs are
outside, each femtocell (both the BS and UE) is inside a house, with different
femtocells in different houses.
We have chosen the lower wall loss values, since
higher values would result
in a more favorable case for the proposed algorithm as the
separation between the femto and macro would be greater.

\begin{table}
\begin{center}
\begin{tabular}{|p{1in}|p{2in}|}
\hline
Link & Path loss (dB) \\ \hline
MBS $\leftrightarrow$ MUE & $15.3 + 37.6\log_{10}(R)$ \\ \hline
MBS $\leftrightarrow$ FUE or FBS & $15.3 + 37.6\log_{10}(R) + L_{ow}$, $L_{ow}=$ 10 \\ \hline
FBS $\leftrightarrow$ FUE (serving link)  & $38.46 + 20\log_{10}(R) + 0.7R$ \\ \hline
FBS or FUE $\leftrightarrow$ FUE or FBS in different femtocell &
$15.3 + 37.6\log_{10}(R) + L_{ow}$, $L_{ow}=$ 20 \\ \hline
FBS or FUE $\leftrightarrow$ MUE &
$15.3 + 37.6\log_{10}(R) + L_{ow}$, $L_{ow}=$ 10 \\ \hline
\end{tabular}
\end{center}
\caption{Path loss models based on \cite{FemtoForum:10}.
$L_{ow}$ is the outer wall loss and $R$ is distance in meters.}
\label{tbl:pathLossMod}
\end{table}

Under these assumptions, we generated a random ``drop" of the femto
and macro network elements and then ran two simulations of
the proposed femto-macro interference control algorithm in Section \ref{sec:powControl}:
one simulation for the DL and a second for the UL.
In each simulation, the load factors in the load-spillage algorithm
were set in an iterative manner as follows.

In the DL, the ``victims" of the femto interference are the macro UEs.
Each macro UEs begins by broadcasting some initial load factors across the subbands.
The femto BS transmitters compute the corresponding spillage terms
and then adjust their power across the subbands to maximize the rate to the femto UE.
The macro UEs then compute the resulting interference and compare the interference
to a target level.  In this simulation, the target level was set to the
interference that would degrade the rate of macro UE by no more than 10\%.
The load factors are then increased or decreased depending on whether
the measured interference is above or below the target.  The algorithm was run
for 50 iterations at which point we observed that more than 99\% of the macro UE experienced an interference level below 0.1 dB of the target.

Similarly, in the UL, the victims of the femto interference are the macro BS receivers.
In this case, the load factors were adjusted to bound the interference rise from
the femtos to 0.5 dB.  This noise rise corresponds to
an approximately 10\% loss in rate in the power-limited regime.

\begin{figure}
\begin{center}
  \includegraphics[height=2.5in]{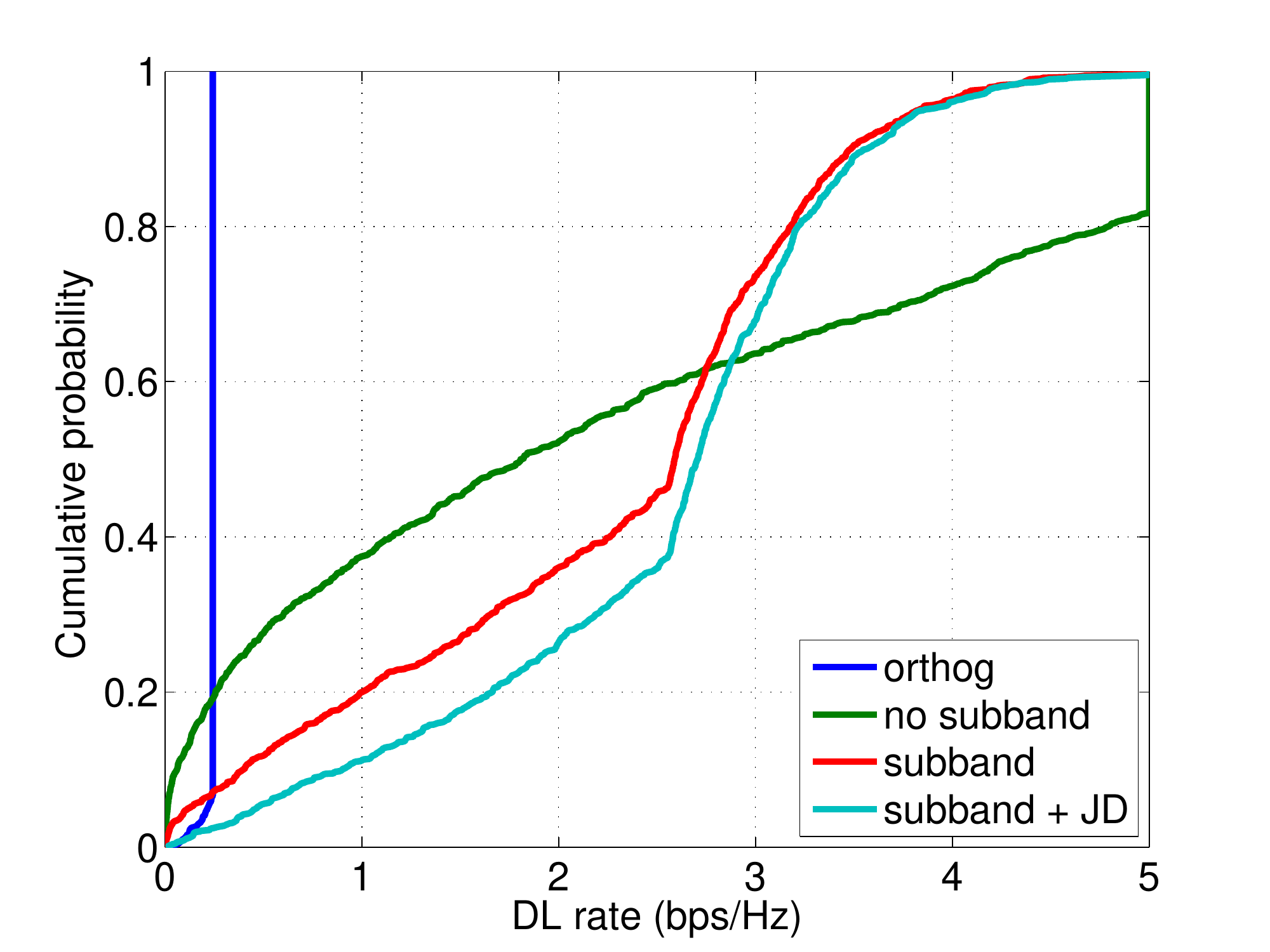}

  \includegraphics[height=2.5in]{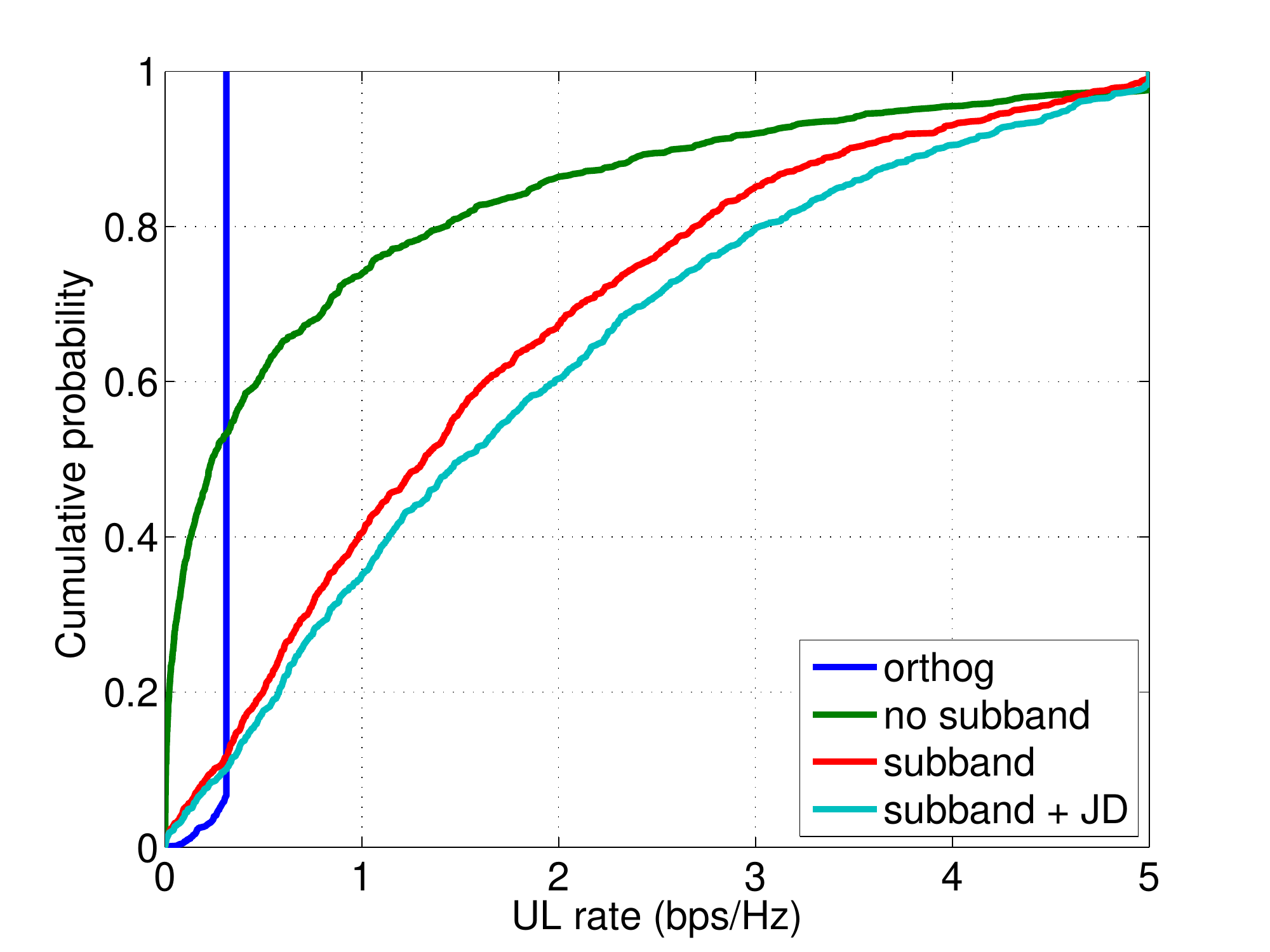}

\end{center}
\caption{Distribution of rates available on short-range links
using subband scheduling with interference cancelation (``subband"),
and subband scheduling with interference cancelation via joint
detection of the desired femto signal and interfering macro signal
(``subband+JD").
Also shown are the rate distributions with simple orthogonalization
(``orthog") and reuse 1 without subband partitioning (``no subband").
Top panel shows the DL and bottom is the UL.
See text for simulation assumptions.
}
\label{fig:specShareSim}
\end{figure}

The distribution of the DL and UL rates of the femto links
after the final iteration of the load-spillage algorithm
are shown in Fig.\ \ref{fig:specShareSim}.
The curves labeled ``subband" are the CDFs without interference
cancelation (IC), while ``subband+JD" includes IC via
joint detection of the macro interfering signal and 
desired femto signal. As a link-layer model for joint detection, 
we assumed the rate achievable via joint detection matched the
Shannon rate with a 3dB loss in both links and maximum spectral 
efficiency of 5 bps/Hz.

The curves show that the load-spillage algorithm is able
to deliver a large rate to a substantial fraction of the femto links.
For example, in the DL, the femto links
achieve a median spectral efficiency of approximately
2.6 bps/Hz without IC. 
Moreover, the loss in macro rate due to femto interference
is small.  Although
the target interference level was set to a maximum 10\% reduction in rate
in the macro, many macro UEs experience an interference well below
the target.  Although not shown in the figure, the actual decrease in
the macro rate observed in the simulation was 4.9\%.

Adding IC via joint detection can improve the rate further,
especially for low rate mobiles.  However, some caveats are in order:
Most importantly, the gains with IC are modest.  Also, our simulation assumed
joint detection.  In simulations of successive IC (SIC) (results are not shown),
SIC alone produced only minimal gains.  This is to be expected since
SIC usually only has benefits for very strong interference.  Since joint detection
is computationally much harder than SIC, the practical value of IC
may be limited.  

Moving to the UL results in Fig.\ \ref{fig:specShareSim}, 
we see that the femto links were able to achieve
median spectral efficiencies of approximately 1.4 bps/Hz without IC
and 1.5 bps/Hz with IC.  The loss in the macro rate was somewhat
higher at 6.2\%, but still small.  The gains with IC in this case
are minimal.  This may be result of the fact that we assumed a rate fair
scheduler per cell in the uplink so no macro links come at very low rate.
Other macro scheduling policies may show further improvements for IC, but
this needs further study.

We conclude that subband scheduling in the macro
combined with the load-spillage power control in the femto can enable a
large rate for many femto links with minimal effect in the macro.
In addition, modest gains are possible in the DL with IC,
especially at low rates.  

Fig.\ \ref{fig:specShareSim} also shows a comparison
of the proposed method to simple orthogonalization.
In orthogonalization, the femtos and macros operate on different
frequencies.  Now, the proposed method resulted in a 4.9\% drop
in the average macro DL spectral efficiency
and an 6.2\% drop in the macro UL.
To compare the proposed method to orthogonalization,
we can subtract that fraction of the bandwidth from the macro
and run the femtos with no macro interference in that bandwidth.
The resulting CDF of the rates are shown in Fig.\ \ref{fig:specShareSim}
on the curve labeled ``orthog".
With orthogonalization from the macro,
many of the femto links get the maximum rate of
5 bps/Hz times the bandwidth fraction.  But, that rate
is much smaller for the overwhelming majority of links.

Fig.\ \ref{fig:specShareSim} also shows the rate distribution
with reuse 1 and power control without subband partitioning.
This can be simulated with the above framework, but with $K=1$
subband.  In the UL, we see that the proposed method offers
significantly greater rate.  In the DL, the median rate is approximately
the same, although the proposed method is more fair in that the
rate variation is less.

\section*{Conclusions}
We have shown that subband partitioning of the macrocell,
in addition to being beneficial to the rate distribution in the macrocell network,
may offer significant opportunities for simultaneous
short-range communication in the same band as the macro.
The key idea to realize this potential rate is
 for the short-range transmitters to
estimate the path losses to the macro receivers and to then
perform some optimal power allocation across the subbands
based on the path loss estimates.
In the simulation scenario considered here,
the proposed method offers significantly greater rates
than simple orthogonalization between the femto and macros.
The proposed method also appears to outperform reuse 1 without subband
 partitioning in the UL.  In the DL, the proposed method has a similar
 median rate, but shows less rate variations.
Further gains are also possible with interference cancelation
in the DL, but requires joint detection and not SIC alone.
The gains in the UL for IC (even with joint detection) are small.

A benefit of the proposed strategy is that
the short-range communication can be achieved
entirely ``under the radar" of the macrocell,
in that there is no need for any explicit communication between
the macros and short-range links.  The only requirements are that
the short-range
transmitters must estimate the path losses to the macro receivers,
the macro receivers must estimate the
aggregate interference rise from the short-range links and broadcast
certain load factors.  The fact that the short-range communication
can share the band with minimal coordination may also be useful in other
scenarios such as cognitive radios where the primary user
needs operate without coordination with potential secondary users.

Further simulations are however needed.  Most importantly,
the current simulations have assumed perfect knowledge
of the path losses and full buffer traffic without any dynamics.
We have also not explicitly modeled the overhead for the signals to
estimate the path losses.

\appendices

\section{Proof of Theorem \ref{thm:powCtrl}}
\label{sec:proof}

To prove (a), fix the load factors $s_{ik} > 0$
and power constraint bounds $\lambda_j > 0$ and
suppose that the powers $\pbf_j$ satisfy \eqref{eq:Uopt}.
To show that the set of $\pbf_j$'s is Pareto optimal,
suppose that there is an alternative set of power vectors $\pbf_j^1$
with $U(\pbf^1_j) \geq U(\pbf_j)$ for all $j$
with strict inequality for at least one $j$.
Let $q^1_{ik}$ be the corresponding interference
levels as in \eqref{eq:qik}.  We must show that $q_{ik}^1 >
q_{ik}$ for some $i$ and $k$.

First observe that since $U_j(\pbf_j)$ is strictly
increasing in each $p_{jk}$, $U_j(\pbf^1_j) \geq U(\pbf_j)$
and $\pbf_j$ is the maxima of \eqref{eq:Uopt} subject to
\eqref{eq:powCons}, we must have that for all $j$,
\[
    \sum_{k=1}^K p_{jk}^1r_{jk} \geq \lambda_j + \epsilon_j, \ \ \
    \sum_{k=1}^K p_{jk}r_{jk} \leq \lambda_j
\]
where $\epsilon_j \geq 0$ for all $j$
with $\epsilon_j > 0$ for at least one $j$.
Therefore,
\beqan
    \lefteqn{ \sum_{k=1}^K p_{jk}^1r_{jk} \geq \sum_{k=1}^K p_{jk}r_{jk}
    + \epsilon_j }\\
    &\stackrel{(a)}{\Leftrightarrow}&
        \sum_{k=1}^K \sum_{i=1}^{M(k)} (p_{jk}^1-p_{jk})s_{ik}G_{ijk} \geq \epsilon_j \\
    &\stackrel{(b)}{\Rightarrow}&
        \sum_{j=1}^N \sum_{k=1}^K \sum_{i=1}^{M(k)} (p_{jk}^1-p_{jk})s_{ik}G_{ijk} > 0 \\
    &\stackrel{(c)}{\Rightarrow}&
        \sum_{k=1}^K \sum_{i=1}^{M(k)} s_{ik}(q_{ik}^1-q_{ik}) > 0
\eeqan
where (a) follows from \eqref{eq:spillage};
(b) follows from taking the sum over $j$ and using the fact that
$\epsilon_j \geq 0$ with strict inequality for at least one $j$;
and
(c) follows from the definition of $q_{ik}$ in \eqref{eq:qik}.
Since $s_{ik} > 0$, we must have that
$q_{ik}^1 > q_{ik}$ for at least one $(i,k)$.
This proves part (a) of Theorem \ref{thm:powCtrl}.

To prove the converse part (b), we need the following standard
result in linear algebra.

\begin{lemma} \label{lem:minMax}  Suppose $\Abf \in \R^{m \x n}$ is
a matrix such that for any $\vbf \in \R^n$,
there exists some component $i=1,\ldots,m$
such that $(\Abf\vbf)_i \leq 0$.  Then, there exist a
$\wbf \in \R^m$  with $\wbf^T\Abf = 0$, $w_i \geq 0$ for all $i$
and $\sum_i w_i = 1$.
\end{lemma}
\begin{proof}  Let $\Wbf$ be the simplex,
\[
    \Wbf := \left\{ \wbf \in \R^m~:~ w_i \geq 0, ~ \sum_i w_i=1 ~ \right\}.
\]
The hypothesis of the lemma implies that for every $\vbf \in \R^n$,
there exists a $\wbf \in \Wbf$ such that $\wbf^T\Abf\vbf \leq 0$.
Therefore,
\[
     \max_{\vbf \in \R^n} \min_{\wbf \in \Wbf} \wbf^T\Abf\vbf = 0.
\]
Since this optimization is convex in $\wbf$ and concave in $\vbf$,
there is a Nash equilibrium \cite{BoydEFB:94}, so we can interchange
the min and max to obtain
\[
     \min_{\wbf \in \Wbf} \max_{\vbf \in \R^n} \wbf^T\Abf\vbf = 0.
\]
But this can occur only if $\wbf^T\Abf = 0$.
\end{proof}

We can use this result to prove the the converse in part (b)
as follows.
Suppose $\pbf_j$, $j=1,\ldots,N$
is a set of Pareto optimal powers.
Since they are Pareto optimal, any small change in the powers
must result in an increase in one of the interference levels $q_{ik}$
or decrease in the utilities $U_j$.
That is, for any non-zero set of $\Delta p_{jk}$'s,
there must either be some $i$ and $k$ such that
\beq \label{eq:deltaq}
    \Delta q_{ik} := \sum_j G_{ijk} \Delta p_{jk} \geq 0,
\eeq
or $j$ such that
\beq \label{eq:deltaU}
    \Delta U_j := \frac{\partial U_j(\pbf_j)}{\partial p_{jk}}
        \Delta p_{jk} \leq 0.
\eeq
Define the matrix $\Abf$ by
\[
    \Abf = \left[ \begin{array}{c}
        -DU(\pbf) \\
        \Gbf
        \end{array}
        \right],
\]
where the columns are indexed by the pairs $(j,k)$,
$DU(\pbf)$ is the matrix with the components
$DU(\pbf)_{j,jk} = \partial U_j(\pbf_j) / \partial p_{jk}$
and $\Gbf$ is the matrix of gains $G_{i,jk}$.
Now, \eqref{eq:deltaq} and \eqref{eq:deltaU} shows that for every
$\Delta \pbf$,
$(\Abf \Delta \pbf)_\ell \leq 0$ for some index $\ell$.
Hence, by Lemma \ref{lem:minMax}, there must exists a $\wbf \geq 0$
with unit norm such that $\wbf^T \Abf = 0$.  Partitioning
\[
    \wbf = [\tbf^T \ \sbf^T]^T
\]
conformably with $\Abf$ we see that there exists
a set of nonnegative $t_j$ and $s_{ik}$ such that
\beq \label{eq:derivUts}
    t_j\frac{\partial U_j(\pbf_j)}{\partial p_{jk}} -
    \sum_{i=1}^{M(k)} s_{ik}G_{ijk} = 0.
\eeq
Now $t_j \geq 0$ for all $j$ and $s_{ik} \geq 0$ for all $i$ and $k$,
with at least one of the parameters being strictly greatly than zero.
Since $G_{ijk} > 0$ and $\partial U_j/\partial p_{jk} > 0$ for all $i$,
$j$ and $k$,
it can be verified that \eqref{eq:derivUts} implies that
that $t_j > 0$ for all $j$.

Now, using \eqref{eq:spillage}, \eqref{eq:derivUts} can be rewritten as
\beq \label{eq:derivL}
    \frac{\partial L_j(\pbf_j)}{\partial p_{jk}} = 0,
\eeq
where $L_j$ is the Lagrangian
\[
    L_j(\pbf_j) := t_jU_j(\pbf_j) - \sum_{k=1}^K r_{jk}p_{jk}.
\]
Let $\lambda_j = \sum_k r_{jk}p_{jk}$.
Since $t_j > 0$, \eqref{eq:derivL} shows that $\pbf_j$ is a critical point
of the optimization \eqref{eq:Uopt} subject to \eqref{eq:powCons}.
Since $U_j(\pbf_j)$ is concave, any critical point is the maxima.
Thus, we have found $s_{ik}$ and $\lambda_j$ such that
the set of power vectors $\pbf_j$ satisfy \eqref{eq:Uopt}.
\bibliographystyle{IEEEtran}
\bibliography{bibl}

\end{document}